\documentclass[a4paper,english]{lipics-v2016}
\usepackage[utf8]{inputenc}
\usepackage{amsfonts}
\usepackage{amssymb}
\usepackage{amsmath}
\usepackage{amsthm}
\usepackage{graphicx}
\usepackage{tabularx}
\usepackage{booktabs}
\usepackage{xspace}
\usepackage[usenames,dvipsnames,svgnames,table]{xcolor}
 
\usepackage{microtype}

\usepackage{tikz}
\usetikzlibrary{calc,positioning,decorations.pathreplacing,fit}
\DeclareUnicodeCharacter{FEFF}{}

\def\ve#1{\mathchoice{\mbox{\boldmath$\displaystyle\bf#1$}}
{\mbox{\boldmath$\textstyle\bf#1$}}
{\mbox{\boldmath$\scriptstyle\bf#1$}}
{\mbox{\boldmath$\scriptscriptstyle\bf#1$}}}
\newcommand\vea{{\ve a}}

\newcommand\veb{{\ve b}}
\newcommand\vecc{{\ve c}}
\newcommand\ved{{\ve d}}

\newcommand\veg{{\ve g}}
\newcommand\veh{{\ve h}}
\newcommand\vel{{\ve l}}

\newcommand\veu{{\ve u}}
\newcommand\vev{{\ve v}}
\newcommand\vew{{\ve w}}
\newcommand\vex{{\ve x}}
\newcommand\vey{{\ve y}}
\newcommand\vez{{\ve z}}

\newcommand{\hy}{\hbox{-}\nobreak\hskip0pt}

\makeatletter
\newtheorem*{rep@theorem}{\rep@title}
\newcommand{\newreptheorem}[2]{%
\newenvironment{rep#1}[1]{%
 \def\rep@title{#2 \ref{##1}}%
 \begin{rep@theorem}}%
 {\end{rep@theorem}}}
\makeatother

\newcommand{\ignore}[1]{}

\newtheorem{proposition}[theorem]{Proposition}
\newreptheorem{theorem}{Theorem}
\newreptheorem{corollary}{Corollary}
\newreptheorem{lemma}{Lemma}

\newcommand{\bt}[1]{\begin{theorem}\label{#1}}
\newcommand{\bc}[1]{\begin{corollary}\label{#1}}
\newcommand{\bl}[1]{\begin{lemma}\label{#1}}
\newcommand{\bp}[1]{\begin{proposition}\label{#1}}
\newcommand{\be}[1]{\begin{example}\rm\label{#1}}
\newcommand{\ba}[1]{\begin{algorithm}\rm\label{#1}}
\newcommand{\bd}[1]{\begin{definition}\rm\label{#1}}
\newcommand{\bpr}{\begin{proof}}
\newcommand{\epr}{\end{proof}}
\newcommand{\et}{\end{theorem}}
\newcommand{\ec}{\end{corollary}}
\newcommand{\el}{\end{lemma}}
\newcommand{\ep}{\end{proposition}}
\newcommand{\ee}{\end{example}}
\newcommand{\ea}{\end{algorithm}}
\newcommand{\ed}{\end{definition}}

\def\N{\mathbb{N}}
\def\R{\mathbb{R}}
\def\Z{\mathbb{Z}}

\def \G {{\cal G}}
\def \T {{\cal T}}
\def \l {\langle}
\def \r {\rangle}

\def \poly {{\rm poly}}
\def \suppo {{\rm supp}}

\def \G {{\cal G}}

\def \C {{\cal C}}

\def \l {\langle}
\def \r {\rangle}

\def \sign {{\rm sign}}
\def \tw {{\rm tw}}
\def \td {{\rm td}}
\def \cl {{\rm cl}}
\def \transpose {{\intercal}}

\newcommand{\NP}{$\mathsf{NP}$\xspace}
\newcommand{\NPh}{$\mathsf{NP}$\hy\textsf{hard}\xspace}
\newcommand{\FPT}{$\mathsf{FPT}$\xspace}

\newcommand{\Wh}[1]{$\mathsf{W[#1]}$\hy\textsf{hard}\xspace}

\usepackage{ifthen}
\newcounter{Accumulate} 
\ifthenelse{\value{Accumulate} = 1}{
  \newwrite\accuwrite \immediate\openout\accuwrite=\jobname.acc
}{}
\usepackage{environ}
\makeatletter
\newenvironment{accumulate}{\Collect@Body\accuPrint}{}
\makeatother
\newcommand{\accuPrint}[1]{
 \ifthenelse{\value{Accumulate} = 0}{%
      #1
  }
  {
    \newtoks\prxxxm
    \prxxxm{#1}
    \immediate\write\accuwrite{\the\prxxxm}
  }
}
\newcommand{\ifaccumulating}[1]{%
  \ifthenelse{\value{Accumulate} = 1}{%
    #1%
  }{}%
}
\newcommand{\ifnoaccumulating}[1]{%
  \ifthenelse{\value{Accumulate} = 0}{%
    #1%
  }{}%
}
\newcommand{\accuprint}{%
  \ifthenelse{\value{Accumulate} = 1}{
    \immediate\closeout\accuwrite
    \input{\jobname.acc} %
  }{}
}
\ifthenelse{\value{Accumulate} = 0}{\let\APXmark\relax}
	{\def\APXmark{{\bf$\!\!$(*)~}}}

\title{A Parameterized Strongly Polynomial Algorithm for Block Structured Integer Programs\footnote{M. Koutecký supported by a postdoctoral fellowship at the Technion. A. Levin supported by  a grant from the GIF, the German-Israeli Foundation for Scientific Research and Development (grant number I-1366-407.6/2016). S. Onn supported by the Dresner chair.}}
\titlerunning{A Parameterized Strongly Polynomial Algorithm for Block Structured ILPs}

\author[1]{Martin Koutecký}
\author[1]{Asaf Levin}
\author[1]{Shmuel Onn}
\affil[1]{Technion -- Israel Institute of Technology, Haifa, Israel \texttt{\{koutecky,levinas,onn\}@technion.ac.il}}

\authorrunning{M. Koutecký, A. Levin, and S. Onn}

\Copyright{Martin Koutecký, Asaf Levin, and Shmuel Onn}
\subjclass{F.2.2 Nonnumerical Algorithms and Problems, G.1.6 Optimization}
\keywords{integer programming, parameterized complexity, Graver basis, $n$-fold integer programming}

\EventEditors{John Q. Open and Joan R. Acces}
\EventNoEds{2}
\EventLongTitle{42nd Conference on Very Important Topics (CVIT 2016)}
\EventShortTitle{CVIT 2016}
\EventAcronym{CVIT}
\EventYear{2016}
\EventDate{December 24--27, 2016}
\EventLocation{Little Whinging, United Kingdom}
\EventLogo{}
\SeriesVolume{42}
\ArticleNo{23}

\begin{document}
\maketitle

\begin{abstract}
The theory of $n$-fold integer programming has been recently emerging as an important
tool in parameterized complexity.
The input to an $n$-fold integer program (IP) consists of parameter $A$, dimension $n$, and numerical data of binary encoding length $L$.
It was known for some time that such programs can be solved in polynomial time using
$O(n^{g(A)}L)$ arithmetic operations where $g$ is an exponential function of the
parameter.
In 2013 it was shown that it can be solved in fixed-parameter tractable time using $O(f(A)n^3L)$
arithmetic operations for a single-exponential function $f$.
This, and a faster algorithm for a special case of \emph{combinatorial} $n$-fold IP, have led to several very recent breakthroughs in the parameterized complexity of scheduling, stringology, and computational social choice.
In 2015 it was shown that it can be solved in strongly polynomial time using $O(n^{g(A)})$ arithmetic operations.

Here we establish a result which subsumes all three of the above results by showing that $n$-fold IP can be solved in strongly polynomial fixed-parameter tractable time using $O(f(A)n^3)$ arithmetic operations.
In fact, our results are much more general, briefly outlined as follows.
\begin{itemize}
\item There is a strongly polynomial algorithm for integer linear programming (ILP) whenever a so-called Graver-best oracle is realizable for it.
\item Graver-best oracles for the large classes of multi-stage stochastic and tree-fold ILPs can be realized in fixed-parameter tractable time.
Together with the previous oracle algorithm, this newly shows two large classes of ILP to be strongly polynomial; in contrast, only few classes of ILP were previously known to be strongly polynomial.
\item We show that ILP is fixed-parameter tractable parameterized by the largest coefficient $\|A\|_\infty$ and the primal or dual treedepth of $A$, and that this parameterization cannot be relaxed, signifying substantial progress in understanding the parameterized complexity of ILP.
\end{itemize}
\end{abstract}

\section{Introduction}

In this article we consider the general linear integer programming (ILP) problem in standard form,
\begin{equation} \label{IP}
  \min\left\{\vew \vex \, \mid A\vex=\veb\,,\ \vel\leq\vex\leq\veu\,,\ \vex\in\Z^{n}\right\} \tag{ILP}.
\end{equation}
with $A$ an integer $m\times n$ matrix, $\veb\in\Z^m$, $\vew\in\Z^n$, $\vel,\veu\in(\Z\cup\{\pm\infty\})^n$.
It is well known to be strongly \NP-hard, which motivates the search for tractable special cases.

The first important special case is ILP in fixed dimension.
In the '80s it was shown by Lenstra and Kannan~\cite{Kannan:1987,Lenstra:1983} that~\eqref{IP} can be solved in time~$n^{O(n)} L$, where $L$ is the length of the binary encoding of the input.
Secondly, it is known that if the matrix $A$ is totally unimodular (all subdeterminants between $-1$ and $1$), all vertices of the feasible region are integral and thus applying any polynomial algorithm for linear programming (LP) suffices.
Later, Veselov and Chirkov~\cite{VeselovC:2009} have shown that the more general class of bimodular ILP is also polynomial-time solvable.
Other results exploit certain structural properties of $A$.
These include the large classes of $n$-fold~\cite{HOR}, tree-fold~\cite{MC}, $2$-stage and multi-stage stochastic~\cite{AH}, and $4$-block $n$-fold~\cite{HKW} ILPs, as well as algorithms for ILPs with bounded treewidth~\cite{GOR}, treedepth~\cite{GO} and fracture number~\cite{DvorakEGKO:2017} of certain graphs related to the matrix $A$.

A fundamental question regarding problems involving large numbers is whether there exists an algorithm whose number of arithmetic operations does not depend on the length of the numbers involved; if this number is polynomial, this is a \emph{strongly polynomial algorithm}~\cite{Tar}.
For example, the ellipsoid method or the interior-point method which solve LP take time which does depend on the encoding length, and the existence of a strongly polynomial algorithm for LP remains a major open problem.
So far, the only strongly polynomial ILP algorithms we are aware of exist for totally unimodular ILP~\cite{HS}, bimodular ILP~\cite{ArtmannWZ:2017}, so-called binet ILP~\cite{AppaKP:2007}, and $n$-fold IP with constant block dimensions~\cite{DeLoeraHL:2015}.
All remaining results, such as Lenstra's famous algorithm or the fixed-parameter tractable algorithm for $n$-fold IP which has recently led to several breakthroughs~\cite{MC,JansenKMR:2018,KnopKM:2017b,KnopKM:2017}, are not strongly polynomial.

\subsection{Our Contributions}
To clearly state our results we introduce the following terminology.
The input to a problem will be partitioned into three parts
$(\alpha, \beta, \gamma)$, where $\alpha$ is the
{\em parametric input}, $\beta$ is the
{\em arithmetic input}, and $\gamma$ is the {\em numeric input}.
A {\em strongly fixed-parameter tractable} (\FPT) \emph{algorithm} for the problem is
one that solves it using $f(\alpha)\poly(\beta)$ arithmetic operations
and $g(\alpha)\poly(\beta,\gamma)$ time, where $f,g$ are some computable functions.
If such an algorithm exists, we say that the problem is \emph{strongly fixed-parameter tractable} (\FPT) \emph{parameterized by $\alpha$}.
Thus, such an algorithm both demonstrates that the problem is \FPT \emph{parameterized by $\alpha$} because it runs in \FPT time $g(\alpha) \poly(\beta, \gamma)$, and provides a strongly polynomial algorithm for each fixed $\alpha$.
Having multiple parameters $\alpha_1, \dots, \alpha_k$ simultaneously is understood as taking the \emph{aggregate parameter} $\alpha = \alpha_1 + \cdots + \alpha_k$.
If the algorithm involves oracles then the oracle queries are also counted as arithmetic operations and the answers to oracle queries should be polynomial in $(\beta,\gamma)$.
Each part of the input may have several entities, which may be presented in unary or binary, where $\l e\r$ denotes the encoding length of an entity $e$ presented in binary.
For the parametric input the distinction between unary and binary is irrelevant.

ILP abounds in natural parameters: the dimension $n$, number of inequalities $m$, largest coefficient $\|A\|_\infty$, largest right-hand side $\|\veb\|_\infty$, various structural parameters of $A$, etc.
Here, we are interested in algorithms which are both strongly polynomial \emph{and} \FPT.

\medskip

Recently it was shown that, if we have access to the so-called {\em Graver basis} of $A$, the problem~\eqref{IP} is polynomial time solvable even for various nonlinear objective functions~\cite{DHK,Onn}.
We show that all of these results can be extended to be strongly polynomial with only $\l A \r$ as the arithmetic input.

\begin{theorem}\label{thm:oracle}
The problem \eqref{IP} with arithmetic input $\l A\r$ and numeric input $\l \vew,\veb,\vel,\veu \r$, endowed with a Graver-best oracle for $A$, is solvable by a strongly polynomial oracle algorithm.
\end{theorem}
The existence of Graver-best oracles is thus of prime interest.
We show such oracles for the wide classes of multi-stage stochastic and tree-fold ILPs; for precise definitions of these classes cf. Section~\ref{subsec:multistage_treefold}.
See Table~\ref{tab:results} for a summary of improvements over the current state of the art.
\begin{theorem}\label{thm:multistage}
Multi-stage stochastic ILP with blocks $B_1, \dots, B_\tau$, $B_i \in \Z^{l \times n_i}$, is strongly \FPT parameterized by $l + n_1 + \cdots + n_\tau$ and $\|A\|_\infty$.
\end{theorem}
\begin{theorem}\label{thm:treefold}
Tree-fold ILP with blocks $A_1, \dots, A_\tau$, $A_i \in \Z^{r_i \times t}$, is strongly \FPT parameterized by $r_1 + \cdots + r_\tau$ and $\|A\|_\infty$.
\end{theorem}
This improves on the algorithm for tree-fold ILP~\cite{MC} not only by making it strongly \FPT, but also by leaving the block length $t$ out of the parameter.
Similarly, the following algorithm for the special case of $n$-fold ILP greatly improves both on the previous results of Hemmecke et al.~\cite{HOR} and Knop et al.~\cite{KnopKM:2017b} and is the currently fastest algorithm for this problem:
\begin{theorem}\label{thm:nfolds}
$n$-fold ILP with blocks $A_1 \in \Z^{r \times t}$ and $A_2 \in \Z^{s \times t}$ can be solved in time $a^{O(r^2s + rs^2)} (nt)^3 + \mathcal{L}(\l A \r)$, where $\mathcal{L}(\l A \r)$ is the runtime of a strongly polynomial LP algorithm.
\end{theorem}

Next, we turn our attention to structural parameters of the constraint matrix $A$.
We focus on two graphs which can be associated with $A$:
\begin{itemize}
\item the \emph{primal graph} $G_P(A)$, which has a vertex for each column and two vertices are connected if there exists a row such that both columns are non-zero, and,
\item the \emph{dual graph} $G_D(A) = G_P(A^{\transpose})$, which is the above with rows and columns swapped.
\end{itemize}
Two standard parameters of structural sparsity are the \emph{treewidth} (measuring the ``tree-likeness'' of a graph) and the more restrictive \emph{treedepth} (measuring its ``star-likeness'').
We denote the treewidth of $G_P(A)$ and $G_D(A)$ by $\tw_P(A)$ and $\tw_D(A)$; for treedepth we have $\td_P(A)$ and $\td_D(A)$.
Note that bounded treedepth implies bounded treewidth but not vice versa.

We show that ILP parameterized by $\td_P(A) + \|A\|_\infty$ and $\td_D(A) + \|A\|_\infty$ can be reduced to the previously mentioned classes, respectively, implying~\eqref{IP} with these parameters is strongly \FPT.

\begin{theorem} \label{thm:primaltd}
\eqref{IP} is strongly \FPT parameterized by $\td_P(A)$ and $\|A\|_\infty$.
\end{theorem}
This improves in two ways upon the result of Ganian and Ordyniak~\cite{GO} who show that \eqref{IP} with $\vew \equiv \mathbf{0}$ (i.e. deciding the feasibility) is \FPT parameterized by $\td_P(A) + \|A, \veb\|_\infty$~\cite{GO}.
First, we use the smaller parameter $\|A\|_\infty$ instead of $\|A, \veb\|_\infty$, and second, we solve not only the feasibility but also the optimization problem.
An analogous results holds for the parameter $\td_D(A)$, for which previously nothing was known at all.
\begin{theorem} \label{thm:dualtd}
\eqref{IP} is strongly \FPT parameterized by $\td_D(A)$ and $\|A\|_\infty$.
\end{theorem}
We emphasize that the parameterizations cannot be relaxed neither from treedepth to treewidth, nor by removing the parameter $\|A\|_\infty$: \eqref{IP} is \NPh already on instances with $\tw_P(A) = 3$ and $\|A\|_\infty = 2$~\cite[Thm 12]{GO}, and it is strongly \Wh{1} parameterized by $\td_P(A)$ alone~\cite[Thm 11]{GO}; the fact that a problem is \Wh{1} is strong evidence that it is not \FPT.
Similarly, deciding feasibility is \NPh on instances with $\tw_D(A) = 3$ and $\|A\|_\infty = 2$ (Lemma~\ref{lem:dualtd_nph}) and strongly \Wh{1} parameterized by $\td_D(A)$ alone~\cite[Thm 5]{KnopKM:2017}.

\begin{table}[tb]
  \centering
  \resizebox{\textwidth}{!}{%
    \begin{tabular}{llr}
      \toprule
      Type of instance             & Previous best run time & Our result\\
      \toprule
      $n$-fold ILP          & $a^{O(rst + st^2)} n^3 L$~\cite{HOR} &  \\
      \cmidrule{1-2}
      $n$-fold ILP          & $t^{O(r)} (ar)^{r^2} n^3 L$ if $A_2 = (1~1~\cdots~1)$~\cite{KnopKM:2017b} & $a^{O(r^2s + sr^2)} (nt)^3 + \mathcal{L}(\l A \r)$~Thm~\ref{thm:nfolds} \\
      \cmidrule{1-2}
      $n$-fold ILP          & $n^{f_1(a,r,s,t)}$~\cite{DeLoeraHL:2015} & \\
      \midrule
      tree-fold ILP       & $f_{\textrm{tf}'}(a,n_1,\dots,n_{\tau},t) n^3 L$~\cite{MC} & $f_{\textrm{tf}}(a,n_1,\dots,n_{\tau}) (nt)^3 + \mathcal{L}(\l A \r)$~Thm~\ref{thm:treefold}\\
      \midrule
      Multi-stage stochastic ILP & $f_{\textrm{mss}}(a, n_1, \dots, n_{\tau}, l) n^3 L$~\cite{AH} & $f_{\textrm{mss}}(a, n_1, \dots, n_{\tau}, l) n^3 + \mathcal{L}(\l A \r)$~Thm~\ref{thm:multistage} \\
	  \midrule
      Bounded dual treedepth & Open whether fixed-parameter tractable & $f_D(a,\td_D(A) (nt)^3 + \mathcal{L}(\l A \r) $~Thm~\ref{thm:dualtd} \\
      \midrule
      Bounded primal treedepth & $f_{P'}(a, \|\veb\|_\infty, \td_P(A)) n L$~\cite{GO} & $f_P(a, \td_P(A)) n^3 + \mathcal{L}(\l A \r)$~Thm~\ref{thm:primaltd} \\
      \bottomrule
    \end{tabular}}
    \caption{Run time improvements implied by this paper. We denote by $L$ the binary length of the numeric input $\veb, \vel, \veu, \vew$, i.e., $L=\l \veb, \vel, \veu, \vew \r$, and consider $\l A \r$ to be part of the arithmetic input. We denote by $a = \max \{2, \|A\|_\infty\}$,$r,s,t$ are relevant block dimensions (cf. Section~\ref{sec:oracles:prelims}), and $\mathcal{L}(\l A \r)$ is the runtime of a strongly polynomial LP algorithm~\cite{Tar}.}
    \label{tab:results}
\end{table}

\subsection{Interpretation of Results}
We believe our approach also leads to several novel insights.
First, we make it clear that the central question is finding Graver-best oracles; provided these oracles, Theorem~\ref{thm:oracle} shows that tasks such as optimization and finding initial solutions can be handled under very mild assumptions.
Even though we show these tasks are routine, they have been reimplemented repeatedly~\cite{MC,HKW,HOR,KnopKM:2017b}.

Second, we show that the special classes of highly uniform block structured ILPs, namely multi-stage stochastic and tree-fold ILPs, are in some sense \emph{universal} for all ILPs of bounded primal or dual treedepth, respectively.
Specifically, we show that any ILP with bounded primal or dual treedepth can be embedded in an equivalent multi-stage stochastic or tree-fold ILP, respectively~(Lemmas~\ref{lem:primaltd_uniformization} and~\ref{lem:dualtd_uniformization}).

Third, we show that, besides bounded primal or dual treedepth, the crucial property for efficiency is the existence of augmenting steps with bounded $\ell_\infty$- or $\ell_1$-norms, respectively~(Lemmas~\ref{lem:multistage_norm} and~\ref{lem:treefold_norm}).
This suggests that for ILPs whose primal or dual graph is somehow ``sparse'' and ``shallow'', finding augmenting steps of bounded $\ell_\infty$- or $\ell_1$-norm might be both sufficient for reaching the optimum and computationally efficient.

\subsection{Related Work}
We have already covered all relevant work regarding strongly polynomial algorithms for ILP.

Let us focus on structural parameterizations.
It follows from Freuder's algorithm~\cite{Freu} and was reproven by Jansen and Kratsch~\cite{JK} that~\eqref{IP} is \FPT parameterized by $\tw_P(A)$ and the largest domain $\|\veu-\vel\|_\infty$.
Regarding the dual graph $G_D(A)$, the parameters $\td_D(A)$ and $\tw_D(A)$ were only recently considered by Ganian et al.~\cite{GOR}.
They show that even deciding feasibility of~\eqref{IP} is \NPh on instances with $\tw_I(A) = 3$ ($\tw_I(A)$ denotes the treewidth of the \emph{incidence graph}; $\tw_I(A) \leq \tw_D(A) + 1$ always holds) and $\|A\|_\infty = 2$~\cite[Theorem 12]{GOR}.
Furthermore, they show that~\eqref{IP} is \FPT parameterized by $\tw_I(A)$ and parameter $\Gamma$, which is an upper bound on any prefix sum of $A \vex$ for any feasible solution $\vex$.

Dvořák et al~\cite{DvorakEGKO:2017} introduce the  parameter fracture number; having a bounded \emph{variable fracture number} $\mathfrak{p}^V(A)$ implies that deleting a few columns of $A$ breaks it into independent blocks of small size; similarly for \emph{constraint fracture number} $\mathfrak{p}^C(A)$ and deleting a few rows.
Because bounded $\mathfrak{p}^V(A)$ implies bounded $\td_P(A)$ and bounded $\mathfrak{p}^C(A)$ implies bounded $\td_D(A)$, our results generalize theirs.
The remaining case of \emph{mixed fracture number} $\mathfrak{p}(A)$, where deleting both rows and columns is allowed, reduces to the $4$-block $n$-fold ILP problem, which is not known to be either \FPT or \Wh{1}.
Because bounded $\mathfrak{p}(A)$ implies bounded $\td_I(A)$, ILP parameterized by $\td_I(A) + \|A\|_\infty$ is at least as hard as $4$-block $n$-fold ILP, highlighting its status as an important open problem.

\subparagraph{Organization.}
The paper contains three main parts.
In Section~\ref{sec:algo}, we provide the proof of Theorem~\ref{thm:oracle}, showing the existence of a strongly polynomial algorithm whenever a Graver-best oracle is provided.
Then, in Section~\ref{sec:multistage_treefold}, we provide Graver-best oracles for multi-stage stochastic and tree-fold ILPs and discuss $n$-fold ILP, and prove Theorems~\ref{thm:multistage}, \ref{thm:treefold} and~\ref{thm:nfolds}.
Finally, in Section~\ref{sec:treedepth} we show how to embed any instance of bounded primal or dual treedepth into a multi-stage stochastic or tree-fold ILP without increasing the relevant parameters, proving Theorems~\ref{thm:primaltd} and~\ref{thm:dualtd}.
\ifthenelse{\value{Accumulate} = 1}
{Due to space restrictions the proofs of our technical statement and other supplementary material are moved to the Appendix.
The statements with proofs presented in the Appendix are marked
with~~\APXmark.}
{}

\section{The Graver-best Oracle Algorithm} \label{sec:algo}
\subsection{Preliminaries} \label{sec:algo:prelims}
For positive integers $m,n$, $m \leq n$, we set $[m,n] = \{m,\ldots, n\}$ and $[n] = [1,n]$.
We write vectors in boldface (e.g., $\vex, \vey$) and their entries in normal font (e.g., the $i$-th entry of~$\vex$ is~$x_i$).
If~$A$ is a matrix, $A_r$ denotes its $r$-th column.
For an integer $a \in \Z$, we denote by $\l a \r = 1 + \log_2 a$ the binary encoding length of $a$; we extend this notation to vectors, matrices and tuples of these objects.
For example, $\l A, \veb \r = \l A \r + \l \veb \r$, and $\l A \r = \sum_{i,j} \l a_{ij} \r$.
For a graph~$G$ we denote by $V(G)$ its set of vertices.
%


\subparagraph{Graver bases and augmentation.}
Let us now introduce Graver bases and discuss how they are used for optimization.
We define a partial order $\sqsubseteq$ on $\R^n$ as follows: for $\vex,\vey\in\R^n$ we write $\vex\sqsubseteq \vey$ and say that $\vex$ is
{\em conformal} to $\vey$ if $x_iy_i\geq 0$ (that is, $\vex$ and $\vey$ lie in the same orthant) and $|x_i|\leq |y_i|$ for $i \in [n]$.
It is well known that every subset of $\Z^n$ has finitely many
$\sqsubseteq$-minimal elements.

\begin{definition}[Graver basis]\label{def:graver}
The {\em Graver basis} of an integer $m\times n$ matrix $A$ is the finite set $\G(A)\subset\Z^n$ of $\sqsubseteq$-minimal elements in $\{\vex\in\Z^n\,:\, A\vex=0,\ \vex\neq 0\}$.
\end{definition}

We say that $\vex$ is \emph{feasible} for~\eqref{IP} if $A\vex = \veb$ and $\vel \leq \vex \leq \veu$.
Let $\vex$ be a feasible solution for~\eqref{IP}.
We call $\veg$ a \emph{feasible step} if $\vex + \veg$ is feasible for~\eqref{IP}.
Further, call a feasible step $\veg$ \emph{augmenting} if $\vew(\vex+\veg) < \vew(\vex)$.
An augmenting step $\veg$ and a \emph{step length} $\alpha \in \Z$ form an \emph{$\vex$\hy{}feasible step pair} with respect to a feasible solution $\vex$ if $\vel \le \vex + \alpha\veg \le \veu$.
An augmenting step $\veh$ is a \emph{Graver-best step} for $\vex$ if $\vew(\vex + \veh) \leq \vew(\vex + \lambda \veg)$ for all $\vex$\hy{}feasible step pairs $(\veg,\lambda) \in \G(A)\times \Z$.
The \emph{Graver\hy{}best augmentation procedure} for~\eqref{IP} with a given feasible solution $\vex_0$ works as follows:
\begin{enumerate}
  \item \label{step1}If there is no Graver\hy{}best step for $\vex_0$, return it as optimal.
  \item If a Graver\hy{}best step $\veh$ for $\vex_0$ exists, set $\vex_0 := \vex_0 + \veh$ and go to \ref{step1}. 
\end{enumerate}
\begin{proposition}[{\cite[Lemma 3.10]{Onn}}]
\label{prop:graverbest}
  Given a feasible solution~$\vex_0$ for~\eqref{IP}, the Graver\hy{}best augmentation procedure finds an optimum of~\eqref{IP} in at most $(2n-2) \log F$ steps, where $F = \vew\vex_0 - \vew\vex^*$ and $\vex^*$ is any minimizer of $\vew \vex$.
\end{proposition}

\begin{definition}[Graver-best oracle] \label{def:gb_oracle}
A {\em Graver-best oracle} for an integer matrix $A$ is one that, queried on $\vew,\veb,\vel,\veu$
and $\vex$ feasible to~\eqref{IP}, returns a Graver-best step $\veh$ for $\vex$.
\end{definition}

\subsection{The Algorithm}
It follows from Proposition~\ref{prop:graverbest} that given a Graver-best oracle, problem \eqref{IP} can be solved in time which is polynomial in the binary encoding length $\l A,\vew,\veb,\vel,\veu \r$
of the input. We now show that, in fact, given such an oracle, the problem admits a
{\em strongly} polynomial algorithm. In the next theorem the input has only arithmetic and numeric parts and no parametric part.

\begin{reptheorem}{thm:oracle}
The problem \eqref{IP} with arithmetic input $\l A\r$ and numeric input $\l \vew,\veb,\vel,\veu \r$, endowed with a Graver-best oracle for $A$, is solvable by a strongly polynomial oracle algorithm.
\end{reptheorem}

\begin{remark}
The partition of the input to the arithmetic input $\l A \r$ and the numeric input $\l \vew, \veb, \vel, \veu \r$ is the same as in the classical results for linear programming~\cite{FT,Tar}.
\end{remark}

\begin{proof}
The algorithm which demonstrates the theorem consists of several steps as follows.

\subparagraph{Step 1: Reducing $\veb,\vel,\veu$.}
Apply the strongly polynomial algorithm of Tardos~\cite{Tar} to the linear programming relaxation $\min\left\{\vew\vey \mid \vey\in\R^n,\, A\vey=\veb, \,\vel\leq \vey\leq \veu\right\}$; the algorithm performs $\mathcal{L}(\l A \r) = \poly(\l A \r)$ arithmetic operations.
If the relaxation is infeasible then so is \eqref{IP} and we are done. If it is unbounded then \eqref{IP}
is either infeasible or unbounded too, and in this case we set $\vew:= \mathbf{0}$ so that all solutions
are optimal, and we proceed as below and terminate at the end of step 3.
Suppose then that we obtain an optimal solution $\vey^*\in\R^n$ to the relaxation,
with round down $\lfloor \vey^*\rfloor\in\Z^n$. Let $a:=\max \{2, \|A\|_\infty \}$.
Let $\C(A)\subseteq\G(A)$ be the set of {\em circuits} of $A$, which are those
$\vecc\in\G(A)$ with support which is a circuit of the linear matroid of $A$.
Let $c_\infty:=\max_{\vecc\in\C(A)}\|\vecc\|_\infty$.
We have $c_\infty\leq n^{n\over2}a^n$~\cite[Lemma 3.18]{Onn}.

We now use the proximity results of \cite{HKW,HS} which assert that either \eqref{IP} is
infeasible or it has an optimal solution $\vex^*$ with $\|\vex^*-\vey^*\|_\infty\leq nc_\infty$ and hence $\|\vex^*-\lfloor \vey^*\rfloor\|_\infty\leq n^{{n\over2}+1}a^n+1$. Thus, making the variable transformation $\vex=\vez+\lfloor \vey^*\rfloor$, problem \eqref{IP} reduces to following,
\begin{equation*}
\min\,\left\{\vew(\vez+\lfloor \vey^*\rfloor) \mid \vez\in\Z^n\,,\ A(\vez+\lfloor \vey^*\rfloor)=\veb
\,,\ \vel\leq \vez+\lfloor \vey^*\rfloor\leq \veu\,,\ \|\vez\|_\infty\leq n^{{n\over2}+1}a^n+1\right\},
\end{equation*}
which is equivalent to the program
\begin{equation}\label{IP1}
\min\,\left\{\vew\vez \mid \vez\in\Z^n\,,\ A\vez=\bar \veb\,,\ \bar \vel\leq \vez\leq \bar \veu\right\}
\end{equation}
where
\[\bar \veb:=\veb-A\lfloor \vey^*\rfloor,\ \
\bar l_i:=\max\{l_i-\lfloor y^*_i\rfloor,-(n^{{n\over2}+1}a^n+1)\},\ \
\bar u_i:=\min\{u_i-\lfloor y^*_i\rfloor,n^{{n\over2}+1}a^n+1\}\ \enspace .\]

If some $\bar l_i>\bar u_i$ then \eqref{IP1} is infeasible
and hence so is \eqref{IP}, so we may assume that
\[-(n^{{n\over2}+1}a^n+1)\leq\bar l_i\leq\bar u_i\leq n^{{n\over2}+1}a^n+1, \quad \mbox{for all }i \enspace .\]

This implies that if $\vez$ is any feasible point in \eqref{IP1} then
$\|A\vez\|_\infty\leq na(n^{{n\over2}+1}a^n+1)$ and so we may assume that
$\|\bar \veb\|_\infty\leq na(n^{{n\over2}+1}a^n+1)$ else there is no feasible solution.
So we have
\[\|\bar \veb\|_\infty,\|\bar \vel\|_\infty,\|\bar \veu\|_\infty \leq 2^{O(n\log n)}a^{O(n)}
\ \mbox{and hence}\ \l\bar \veb,\bar \vel,\bar \veu\r\ \mbox{is polynomial in}\ \l A\r\ \enspace .\]

\subparagraph{Step 2: Solving the system of equations.}
We first search for an integer solution to the system of equations $A\vez=\bar \veb$. This can be done by computing the Hermite normal form of $A$, see \cite{Sch}, using a number of arithmetic operations polynomial in $\l A\r$ and time polynomial in $\l A,\bar \veb\r$ which is polynomial in $\l A\r$, and hence strongly
polynomially in our original input. Then either we conclude that there is no integer solution
to $A\vez=\bar \veb$ and hence \eqref{IP1} is infeasible, or we find a solution $\vez\in\Z^n$ with
$\l \vez\r$ polynomially bounded in $\l A,\bar \veb\r$ and hence also in $\l A\r$.

\subparagraph{Step 3: Finding a feasible point.}
Define relaxed bounds by
\[
{\hat l}_i:=\min\{\bar l_i,z_i\},\ \ {\hat u}_i:=\max\{\bar u_i,z_i\},\quad i \in [n] \enspace .
\]
Now for $i \in [n]$ iterate the following. If $\bar l_i\leq z_i\leq \bar u_i$ then simply
increment $i$ and repeat. If $z_i<\bar l_i$ (and hence $\hat l_i=z_i$ and $\hat u_i=\bar u_i$) then consider the following auxiliary integer program,
\begin{equation}\label{IP2}
\max\,\left\{x_i \mid \vex\in\Z^n\,,\ A\vex=\bar \veb\,,\ \hat \vel\leq \vex\leq \hat \veu\right\}.
\end{equation}
Starting from the point $\vez$ feasible in \eqref{IP2}, and using the Graver-best oracle for $A$, we can solve program \eqref{IP2} using Proposition~\ref{prop:graverbest} in polynomial time and in a number of arithmetic operations and oracle queries which is polynomial in $n$ and $\log F$ (recall $F = z_i^* - z_i$ for some minimizer $z_i^*$), which is bounded by
$\log(\hat u_i-\hat l_i)=\log(\bar u_i-z_i)$, thus polynomial in $\l A\r$.

Let $\vex$ be an optimal solution of \eqref{IP2}. If $x_i<\bar l_i$ then \eqref{IP1}
is infeasible and we are done. Otherwise (in which case $\bar l_i\leq x_i\leq \bar u_i$)
we update $\hat l_i:=\bar l_i$ and $\vez:=\vex$, increment $i$ and repeat. The last case $z_i>\bar u_i$ is treated similarly where in \eqref{IP2} we minimize rather than maximize $x_i$.

Thus, strongly polynomially we either conclude at some iteration $i$
that program \eqref{IP1} is infeasible or complete all iterations and
obtain $\hat \vel=\bar \vel$, $\hat \veu=\bar \veu$, and a point $\vez$ feasible in \eqref{IP1}.

\subparagraph{Step 4: Reducing $\vew$.}
Let $N:=2n(n^{{n\over2}+1}a^n+1)+1$.
Now apply the strongly polynomial algorithm of Frank and Tardos~\cite{FT}, which on arithmetic input $n,\l N\r$ and numeric input $\l \vew\r$, outputs $\bar \vew\in\Z^n$ with $\|\bar \vew\|_\infty \leq 2^{O(n^3)}N^{O(n^2)}$ such that $\sign(\vew\vex)=\sign(\bar \vew\vex)$ for all $\vex\in\Z^n$ with $\|\vex\|_1<N$.
Since $\l N\r=O(\log N)=O(n\log n+ n\log a)$ is polynomial in $\l A\r$, this algorithm is also strongly polynomial in our original input. Now, for every two points $\vex,\vez$ feasible in \eqref{IP1} we have $\|\vex-\vez\|_1<2n(n^{{n\over2}+1}a^n+1)+1=N$, so that for
any two such points we have $\vew\vex\leq \vew\vez$ if and only if $\bar \vew\vex\leq \bar \vew\vez$,
and therefore we can replace \eqref{IP1} by the equivalent program
\begin{equation}\label{IP3}
\min\,\left\{\bar \vew\vez\ :\ \vez\in\Z^n\,,\ A\vez=\bar \veb\,,\ \bar \vel\leq \vez\leq \bar \veu\right\},
\end{equation}
where
\[
\|\bar \vew\|_\infty=2^{O(n^3\log n)}a^{O(n^3)}
\ \mbox{and hence}\ \l\bar \vew,\bar \veb,\bar \vel,\bar \veu\r\ \mbox{is polynomial in}\ \l A\r\ \enspace .
\]

\subparagraph{Step 5: Finding an optimal solution.}
Starting from the point $\vez$ which is feasible in \eqref{IP3}, and using the Graver-best oracle for $A$, we can solve program \eqref{IP3} using again Proposition~\ref{prop:graverbest} in polynomial time and in a number of arithmetic operations and oracle queries which is polynomial in $n$ and in $\log F$, which is bounded by
$\log\left(n\|\bar \vew\|_\infty\|\bar \veu-\bar \vel\|_\infty\right)$,
which is polynomial in $\l A\r$, and hence strongly polynomially.
\epr

\begin{remark}
In fact, the reduced objective $\bar{\vew}$ in step 4 need not be constructed: already its \emph{existence} implies that~\eqref{IP1} is solved in the same number of iterations as~\eqref{IP3}.
\end{remark}

\section{Multi-stage Stochastic and Tree-fold ILP} \label{sec:multistage_treefold}
\begin{accumulate}
\ifthenelse{\value{Accumulate} = 1}{
\section{Additions to Section~\ref{sec:multistage_treefold}}
}{}\end{accumulate}
In this section we prove Theorems~\ref{thm:multistage} and~\ref{thm:treefold}.
We first formalize a common construction for a Graver-best oracle: one constructs a set of relevant step lenghts $\Lambda$ and then for each $\lambda \in \Lambda$ finds a $\lambda$-Graver-best step.
A step with the best improvement among these is then guaranteed to be a Graver-best step.
Thus, we reduce our task to constructing a $\Lambda$-Graver-best oracle.

Both algorithms for multi-stage stochastic ILP and tree-fold ILP follow the same pattern:
\begin{enumerate}
\item show that all elements of $\G(A)$ have bounded norms ($\ell_\infty$ and $\ell_1$, respectively),
\item show that $A$ has bounded treewidth (primal and dual, respectively),
\item apply existing algorithms for~\eqref{IP} which are \FPT parameterized by $\|A\|_\infty$, $\max \|\vex\|_\infty$ and $\max \|\vex\|_1$, and $\tw_P(A)$ and $\tw_D(A)$, respectively.
\end{enumerate}

\subsection{Preliminaries} \label{sec:oracles:prelims}
\subsubsection{Relevant Step Lengths}
We say that $\veh \in \{\vex \in \Z^n\mid A \vex = \mathbf{0} \}$ is a \emph{$\lambda$-Graver-best step} if $\lambda \veh$ is a feasible step and $\lambda \vew \veh \leq \lambda \vew \veg$ for any $\veg \in \G(A)$ such that $\lambda \veg$ is a feasible step.
We denote by $g_1(A) = \max_{\veg \in \G(A)} \|\veg\|_1$ and $g_{\infty}(A) = \max_{\veg \in \G(A)} \|\veg\|_\infty$.
The following lemma states that provided a bound on $g_{\infty}(A)$, in order to find a Graver-best step, it is sufficient to find a $\lambda$-Graver-best step for all $\lambda \in \Lambda$ for some not too large set $\Lambda$.

\begin{definition}[Graver-best step-lengths]
Let $\vex$ be a feasible solution to \eqref{IP}.
We say that $\lambda \in \N$ is a \emph{Graver-best step-length for $\vex$} if there exists $\veg \in \G(A)$ with $\vex + \lambda \veg$ feasible, such that $\forall \lambda' \in \N$ and $\forall \veg' \in \G(A)$, $\vex + \lambda' \veg'$ is either infeasible or $\vew(\vex + \lambda \veg) \leq \vew(\vex + \lambda' \veg')$.
We denote by $\Lambda(\vex) \subseteq \N$ the set of Graver-best step-lengths for $\vex$.
\end{definition}

\begin{lemma}[Polynomial $\Lambda \supseteq \Lambda(\vex)$] \APXmark
\label{lem:gamma_construction}
Let $\vex$ be a feasible solution to \eqref{IP}, let $M \in \N$ be such that $g_\infty(A) \leq M$.
Then it is possible to construct in time $O(Mn)$ a set $\Lambda \subseteq \N$ of size at most $2Mn$ such that $\Lambda(\vex) \subseteq \Lambda$.
\end{lemma}
\begin{accumulate}
\ifthenelse{\value{Accumulate} = 1}{
\begin{proof}[Proof of Lemma~\ref{lem:gamma_construction}]}
{\bpr}
Observe that if $\lambda \veg$ is a Graver-best step for $\vex$, it is tight in at least one coordinate $i \in [n]$, meaning that $l_i \leq x_i + \lambda g_i \leq u_i$ but either $x_i + (\lambda + 1) g_i < l_i$ or $x_i + (\lambda + 1) g_i > u_i$.
Thus, for every $i \in [n]$ and every $\mu \in [-M, M]$, $\mu \neq 0$, add to $\Lambda$ the unique number $\lambda \in \N$ for which $l_i \leq x_i + \lambda \mu \leq u_i$ holds but $l_i \leq x_i + (\lambda+1) \mu \leq u_i$ does not hold.
\epr
\end{accumulate}

With this $\Lambda$ at hand, in order to realize a Graver-best oracle, it suffices to realize an oracle which finds a $\lambda$-Graver-best step for a given $\lambda$:

\begin{definition}[$\Lambda$-Graver-best oracle]
A \emph{$\Lambda$-Graver-best oracle} for an integer matrix $A$ is one that, queried on $\vew, \veb, \vel, \veu$, $\vex$ feasible to~\eqref{IP}, and an integer $\lambda \in \N$, returns a $\lambda$-Graver-best step $\veh$ for $\vex$.
\end{definition}

\begin{lemma}[$\Lambda$-Graver-best oracle $\Rightarrow$ Graver-best oracle] \APXmark \label{lem:lambda_gb_oracle}
Let $A$ be an integer matrix and let $M \in \N$ satisfy $g_\infty(A) \leq M$.
Then a Graver-best oracle for $A$ can be realized with $2Mn$ calls to a $\Lambda$-Graver-best oracle for $A$.
\end{lemma}
\begin{accumulate}
\ifthenelse{\value{Accumulate} = 1}{
\begin{proof}[Proof of Lemma~\ref{lem:lambda_gb_oracle}]}
{\bpr}
Construct the set $\Lambda$ by Lemma~\ref{lem:gamma_construction} and for each $\lambda \in \Lambda$, call the $\Lambda$-Graver-best oracle and denote its output $\veh_\lambda$.
Then, return $\veh_\lambda$ which minimizes $\vew \lambda \veh_\lambda$ over all $\lambda \in \Lambda$.
By Lemma~\ref{lem:gamma_construction} it is guaranteed to be Graver-best step, and since $|\Lambda| \leq 2Mn$ we have made at most this many calls of the $\Lambda$-Graver-best oracle.
\epr
\end{accumulate}

\subsubsection{Multi-stage Stochastic and Tree-fold Matrices} \label{subsec:multistage_treefold}
Let the {\em height} of a rooted tree or forest be the maximum root-to-leaf distance
in it (i.e., the number of edges along the root-to-leaf path).
For a vertex $v \in T$, let $T_v$ be the subtree of $T$ rooted in $v$ and let $\ell(v)$ denote the number of leaves of $T$ contained in $T_v$.
In the following we let $T$ be a rooted tree of height $\tau-1 \in \N$ whose all leaves are at \emph{depth} $\tau-1$, that is, the length of every root-leaf path is exactly $\tau-1$.
Let $B_1,B_2,\dots,B_{\tau}$ be a sequence of integer matrices
with each $B_s$ having $l \in \N$ rows and $n_s$ columns, where $n_s\in\N$, $n_s\geq 1$.
We shall define a \emph{multi-stage stochastic} matrix $T^P(B_1, \dots, B_{\tau})$ inductively; the superscript $P$ refers to the fact that $T^P(B_1, \dots, B_{\tau})$ has bounded \emph{primal} treedepth $\td_P$, as we will later see.

For a leaf $v \in T$, $T^P_v(B_{\tau}) := B_{\tau}$.
Let $d \in \N$, $0 \leq d \leq {\tau-2}$, and assume that for all vertices $v \in T$ at depth $d+1$, matrices $T^P_v(B_{\tau-d}, \dots, B_{\tau})$ have been defined.
For $s \in \N$, $1 \leq s \leq \tau$, we set $T^P_v(B_{[s:\tau]}) = T^P_v(B_s, \dots, B_{\tau})$.
Let $v \in T$ be a vertex at depth $d$ with $\delta$ children $v_1, \dots, v_\delta$.
We set
\[
T^P_v(B_{[\tau-d-1:\tau]}):=\left(
\begin{array}{cccc}
B_{\tau-d-1, \ell(v_1)}      & T^P_{v_1}(B_{[\tau-d:\tau]})        & \cdots & 0\\
\vdots   & \vdots    & \ddots & \vdots\\
B_{\tau-d-1, \ell(v_\delta)}            & 0       & \cdots & T^P_{v_\delta}(B_{[\tau-d:\tau]})\\
\end{array}
\right)
\]
where, for $N \in \N$, $B_{s,N}=\left(\begin{smallmatrix} B_s \\ \vdots \\ B_s\end{smallmatrix}\right)$
consists of $N$ copies of the matrix $B_s$.

The structure of a multi-stage stochastic matrix makes it natural to partition any solution of a multi-stage stochastic ILP into \emph{bricks}.
Bricks are defined inductively: for $T_v^P(B_{\tau})$ there is only one brick consisting of all coordinates; for $T_v^P(B_{[s:\tau]})$ the set of bricks is composed of all bricks for all descendants of $v$, plus the first $n_s$ coordinates form an additional brick. 

\begin{example}
For $\tau=3$ and $T$ with root $r$ of degree $2$ and its children $u$ and $v$ of degree $2$ and $3$, we have $T^P_u(B_2, B_3)=
\left(\begin{smallmatrix}
B_2 & B_3  &     \\
B_2 &      & B_3
\end{smallmatrix}\right)$,
$T^P_v(B_2, B_3)=
\left(\begin{smallmatrix}
B_2 & B_3  &    & \\
B_2 &      & B_3 & \\
B_2 &      &  & B_3
\end{smallmatrix}\right)$, and
$T^P(B_1, B_2, B_2) = T^P_r(B_1, B_2, B_2)=
\left(\begin{smallmatrix}
B_1 & B_2 & B_3 &     &     &     &    & \\
B_1 & B_2 &     & B_3 &     &     &    & \\
B_1 &     &     &     & B_2 & B_3 &    & \\
B_1 &     &     &     & B_2 &     & B_3& \\
B_1 &     &     &     & B_2 &     &    & B_3 \\
\end{smallmatrix}\right)$, with a total of $8$ bricks.
\end{example}


\emph{Tree-fold} matrices are essentially transposes of multi-stage stochastic ILP matrices.
Let $T$ be as before and $A_1, \dots, A_{\tau}$ be a sequence of integer matrices with each $A_s \in \Z^{r_s \times t}$, where $t \in \N$, $r_s \in \N$, $r_s \geq 1$.
We shall define $T^D(A_1, \dots, A_{\tau})$ inductively; the superscript $D$ refers to the fact that $T^D(A_1, \dots, A_{\tau})$ has bounded \emph{dual} treedepth.
The inductive definition is the same as before except that, for a vertex $v \in T$ at depth $d$ with $\delta$ children $v_1, \dots, v_\delta$,
we set
\[
T^D_v(A_{[\tau-d-1:\tau]}):=\left(
\begin{array}{ccccc}
A_{\tau-d-1, \ell(v_1)}      & A_{\tau-d-1, \ell(v_2)}  & \cdots & A_{\tau-d-1, \ell(v_\delta)}   \\
T^D_{v_1}(A_{[\tau-d:\tau]})      & 0       & \cdots & 0\\
0 & T^D_{v_2}(A_{[\tau-d:\tau]})           & \cdots & 0\\
\vdots   & \vdots   & \ddots & \vdots\\
0      & 0 & \cdots & T^D_{v_\delta}(A_{[\tau-d:\tau]})      \\
\end{array}
\right)
\]
where, for $N \in \N$, $A_{s,N}=\left( A_s \\ \cdots \\ A_s\right)$ consists of $N$ copies of the matrix $A_s$.
A solution $\vex$ of a tree-fold ILP is partitioned into bricks $(\vex^1, \dots, \vex^n)$ where $n$ is the number of leaves of $T$, and each $\vex^i$ is a $t$-dimensional vector.

\subsubsection{Structural Parameters}
We consider two graph parameters, namely \emph{treewidth} $\tw(G)$ and \emph{treedepth} $\td(G)$.
\ifthenelse{\value{Accumulate} = 1}{
We postpone the definition of treewidth to the Appendix as it is not central for us.}{}

\begin{accumulate}
\begin{definition}[Treewidth]
A \emph{tree decomposition} of a graph $G=(V,E)$ 
is a pair $(T, B)$, where $T$ is a tree and $B$ is a mapping
$B: V(T) \rightarrow 2^V$ satisfying
\begin{itemize}
	\item for any $uv \in E$, there exists $a \in V(T)$ such that
	$u, v \in B(a)$,
	\item if $v \in B(a)$ and $v \in B(b)$, then $v \in B(c)$ for all
	$c$ on the path from $a$ to $b$ in $T$.
\end{itemize}
We use the convention that the vertices of the tree are called \emph{nodes} and the sets
$B(a)$ are called \emph{bags}.
The {\em treewidth $\tw((T, B))$ of a tree decomposition} $(T, B)$ is
the size of the largest bag of $(T, B)$ minus one.
The {\em treewidth $\tw(G)$ of a graph} $G$ is the
minimum treewidth over all possible tree decompositions of $G$.
\end{definition}
\end{accumulate}

\begin{definition}[Treedepth]
\label{def:tree-depth}
The {\em closure} $\cl(F)$ of a rooted forest $F$
is the graph obtained from $F$ by making every vertex adjacent to all of its
ancestors.
The {\em treedepth} $\td(G)$ of a graph $G$ is one more than the minimum
height of a forest $F$ such that $G\subseteq \cl(F)$.
\end{definition}

It is known that $\tw(G) \leq \td(G)$.
The treedepth $\td(G)$ of a graph $G$ with a witness forest $F$ can be computed in time $f_{\textrm{td}}(\td(G)) \cdot |V(G)|$ for some computable function $f_{\textrm{td}}$~\cite{ReidlRVS:2014}.

\begin{definition}[Primal and dual graph]\label{primaldual-graph}
Given a matrix $A \in \Z^{m \times n}$, its \emph{primal graph} $G_P(A) = (V,E)$ is defined as $V = [n]$ and $E = \{\{i,j\} \in \binom{[n]}{2} \mid \exists k \in [m]: A_{k,i}, A_{k,j} \neq 0\}$.
In other words, its vertices are the columns of $A$ and two vertices are connected if there is a row with non-zero entries at the corresponding columns.
The \emph{dual graph of $A$} is defined as $G_D(A) = G_P(A^\transpose)$, that is, the primal graph of the transpose of $A$.
\end{definition}

\begin{accumulate}
\subparagraph{Graphs of $A$.}
Besides $G_P(A)$ and $G_D(A)$, also the incidence graph is studied.
The \emph{incidence graph of $A$} $G_I(A) = (V_I, E_I)$ is defined as $V_I = \{v_i \mid i \in [n]\} \cup \{c_j \mid j \in [m]\}$ and $E_I = \left\{ \{v_i, c_j\} \mid A_{i,j} \neq 0, \, i \in [n], j \in [m] \right\}$
To gain some intuition for how the primal, dual and incidence graphs are related, consider the following.
For a graph $G$ let $G^2$ denote the \emph{square of $G$} which is obtained from $G$ by adding an edge between all vertices in distance $2$.
For a subset of vertices $W \subseteq V(G)$, we denote by $G[W]$ the subgraph of $G$ induced by $W$.
It is easy to see that $G_P(A) = G_I(A)^2[\{v_i \mid i \in [n]\}]$ and $G_D(A) = G_I(A)^2[\{c_j \mid j \in [m]\}]$.
\end{accumulate}

\begin{definition}[Matrix treewidth]\label{matrix-treewidth}
Given a matrix $A$, its \emph{primal treewidth} $\tw_P(A)$ is defined as the treewidth of its primal graph, i.e., $\tw(G_P(A))$, and its \emph{dual treewidth} $\tw_D(A)$ is $\tw(G_D(A))$.
Similarly, we define the \emph{primal} and \emph{dual treedepth} as $\td_P(A) = \td(G_P(A))$ and $\td_D(A) = \td(G_D(A))$, respectively.
\end{definition}

\begin{accumulate}
We have the following similar bounds which will be useful later.
\begin{lemma}[{Kolaitis and Vardi~\cite{KolaitisV:2000}}] \label{lem:inc_prim_tw}
$\tw_I(A) \leq \tw_P(A) + 1$ and $\tw_I(A) \leq \tw_D(A) + 1$.
\end{lemma}

\bpr
Construct a tree decomposition $T'$ of $G_I(A)$ from an optimal tree decomposition $T$ of $G_P(A)$ as follows.
Consider a row $\vea_i$ of $A$: its non-zeros correspond to a clique of columns in $G_P(A)$, so there must exist a bag of $T$ containing all of them; now add the vertex corresponding to $\vea_i$ to this bag.
Repeating this for all rows and possibly copying bags obtains $T'$ of width at most one larger than $T$.
A similar argument applies for $G_D(A)$, exchanging rows for columns and vice-versa.
\epr

\begin{lemma} \label{lem:inc_prim_td}
For $A \in \Z^{m \times n}$, we have $\td_I(A) \leq \td_P(A) + 1$ and $\td_I(A) \leq \td_D(A) + 1$.
\end{lemma}

\bpr
We shall prove $\td_I(A) \leq \td_P(A) + 1$; the other claim follows by noting that $G_I(A^\transpose) = G_I(A)$.
Let $F$ be a forest of height $\td_P(A) - 1$ such that $G_P(A) \subseteq \cl(F)$.
Every root-leaf path $p$ in $F$ corresponds to a subset $S_p \subseteq [n]$ of columns of $A$.
Let $R_p \subseteq [m]$ be the subset of rows of $A$ whose non-zero entries are contained in $S_p$.
Obtain $F'$ from $F$ by, for every root-leaf path $p$, appending a leaf to the last vertex of $p$ for every $r \in R_p$.
Let us argue that $G_I(A) \subseteq \cl(F')$.
Every non-zero row $r$ belongs to some $R_p$ and thus has a corresponding leaf.
Moreover, $r$ only has edges in $G_I(A)$ connecting it to vertices from $S_p$, and all such edges belong to $\cl(F')$.
Finally, the depth of $F'$ is clearly $1$ plus the height of $F$ as we have only added one more layer of leaves.
\epr
\end{accumulate}

\ifthenelse{\value{Accumulate} = 1}{
Using a proof of Ganian et al.~\cite[Theorem 12]{GOR} we show that we cannot hope to relax the parameter $\td_D(A)$ to $\tw_D(A)$, even if $\|A\|_\infty$ was a constant.
\begin{lemma} \APXmark \label{lem:dualtd_nph}
\eqref{IP} is \NPh already when $\tw_D(A) = 3$, $\|A\|_\infty = 2$, and $\vew = \mathbf{0}$.
\end{lemma}}
{}

\subsection{Multi-stage Stochastic ILP is strongly \FPT}
To prove Theorem~\ref{thm:multistage}, we need two ingredients:
a bound on $g_\infty(A)$, and an algorithm for~\eqref{IP} with bounded $\tw_P(A)$ and $\max \|\vex\|_\infty$.

\begin{lemma}[Multi-stage stochastic $\Rightarrow$ bounded $g_\infty(A)$] \APXmark
\label{lem:multistage_norm}
Let $A = T^P(B_1, \dots, B_\tau)$.
Then $g_\infty(A) \leq f_{\textrm{mss-norm}}(a, n_1, \dots, n_\tau, l)$ for some computable function $f_{\textrm{mss-norm}}$.
\end{lemma}

\begin{accumulate}
\ifthenelse{\value{Accumulate} = 1}{
\begin{proof}[Proof of Lemma~\ref{lem:multistage_norm}]}
{\bpr}
Aschenbrenner and Hemmecke~\cite{AH} show that each brick of $\veg$ comes from a union $\bigcup_{k=1}^\tau \mathcal{H}_{k,\infty}$ which is independent of $n$ and finite~\cite[Proposition 8.11]{AH}.
Thus, $\|\veg\|_\infty \leq f_{\textrm{mss-norm}}(a, n_1, \dots, n_\tau, l)$ for some function $f_{\textrm{mss-norm}}$ which is computable by arguments in~\cite[Section 9]{AH}.
\epr
\end{accumulate}

\begin{lemma} \APXmark \label{lem:primal_treewidth}
Let $X \in \N$.
Problem~\eqref{IP} with the additional constraint $\|\vex\|_\infty \leq X$
can be solved in time $(X+1)^{O(\tw_P(A))} \cdot (n+m)$.
\end{lemma}

\begin{accumulate}
\ifthenelse{\value{Accumulate} = 1}{
\begin{proof}[Proof of Lemma~\ref{lem:primal_treewidth}]}
{\bpr}
The algorithm follows from Freuder's algorithm:
\begin{proposition}[Freuder~{\cite{Freu,JK}}] \label{prop:primal_treewidth}
\eqref{IP} can be solved in time $\|\veu - \vel\|_\infty^{O(\tw_P(A))} \cdot n$.
\end{proposition}
We apply it to~\eqref{IP} with modified bounds $\tilde{\vel}$ and $\tilde{\veu}$ defined by $\tilde{l}_i = \max \{l_i, -X\}$ and $\tilde{u}_i = \min \{u_i, X\}$.
It is clear that $\|\tilde{\veu} - \tilde{\vel}\|_\infty \leq X$ and since $A$ and thus $\tw_P(A)$ stay the same, the claim follows.
\epr
\end{accumulate}

\begin{proof}[Proof of Theorem~\ref{thm:multistage}]
Let $A =  T^P(B_1, \dots, B_\tau)$ be a multi-stage stochastic matrix.
By Lemma~\ref{lem:multistage_norm}, $g_\infty(A)$ is bounded by $M = f_{\textrm{mss-norm}}(a, n_1, \dots, n_\tau, l)$.
We show how to construct a $\Lambda$-Graver-best oracle.
Given an integer $\lambda \in \N$, use Lemma~\ref{lem:primal_treewidth} to solve
$$
\min \{\lambda \vew \veh \mid A \veh = 0, \, \vel \leq \vex + \lambda \veh \leq \veu, \, \|\veh\|_\infty \leq M\} \enspace .
$$
This returns a $\lambda$-Graver-best step, because any optimal solution satisfies $\lambda \veh \leq \lambda \veg$ for all $\veg \in \G(A)$.
Using a simple induction and the inductive construction of $A$, one gets that $A$ has $\tw_P(A) \leq \td_P(A) \leq n_1 + \cdots + n_\tau + 1$ and thus the oracle is realized in \FPT time.
Lemma~\ref{lem:lambda_gb_oracle} then yields a Graver-best oracle, which, combined with Theorem~\ref{thm:oracle}, finishes the proof.
\end{proof}

\subsection{Tree-fold ILP is strongly \FPT}
As before, to prove Theorem~\ref{thm:treefold}, we need two ingredients:
a bound on $g_1(A)$, and an algorithm for~\eqref{IP} with bounded $\tw_D(A)$ and $\max \|\vex\|_1$.
\begin{lemma}[Tree-fold $\Rightarrow$ bounded $g_1(A)$] \APXmark\label{lem:treefold_norm}
Let $A_i \in \Z^{r_i \times t}$ for $i \in [\tau]$ with $a = \max \{ 2, \max_{i \in [\tau]} \|A_i\|_\infty \}$, $r = \sum_{i=1}^\tau r_i$. 
Let $A = T^D(A_1, \dots, A_\tau)$.
There exists a computable function $f_{\textrm{tf-norm}}(a, r_1, \dots, r_\tau)$ such that $g_1(A) \leq f_{\textrm{tf-norm}}(a, r_1, \dots, r_\tau)$.
\end{lemma}

\ifthenelse{\value{Accumulate} = 1}{
\subparagraph{Proof sketch.} Chen and Marx~\cite{MC} prove a similar result under the assumption that $t$ is also a parameter; thus, the remaining problem are essentially duplicitous columns.
However, De Loera et al.~\cite{DHK} show that repeating columns of any matrix $A'$ does not increase $g_1(A')$, and thus we can take $A$, delete duplicitous columns, apply the result of Chen and Marx, and our Lemma follows.
}
{}
\begin{accumulate}
\ifthenelse{\value{Accumulate} = 1}{
\begin{proof}[Proof of Lemma~\ref{lem:treefold_norm}]}
{\bpr}
Let $C = \left(\begin{smallmatrix}A_1 \\ \vdots \\ A_\tau \end{smallmatrix}\right)$ and obtain $\tilde{C} = \left(\begin{smallmatrix}\tilde{A}_1 \\ \vdots \\ \tilde{A}_\tau \end{smallmatrix}\right)$ by deleting duplicitous columns of $C$.
Let $\tilde{A} = T^D(\tilde{A}_1, \dots, \tilde{A}_\tau)$.
This tree-fold matrix (unlike $A$ itself) satisfies the property that the number $\tilde{t}$ of columns of $\tilde{A}_i$, for each $i \in [\tau]$, is bounded by $(2a+1)^r$.

Let $n' \in \N$ be the number of bricks of $A$.
Chen and Marx~\cite[Lemma 11]{MC} show that for every $\tilde{\veg} = (\tilde{\veg}^1, \dots, \tilde{\veg}^{n'}) \in \G(\tilde{A})$ and for every $I \subseteq [n']$, $\sum_{i \in I} \tilde{\veg}^i$ is a sum of at most $f_2(\tilde{A}_1, \dots, \tilde{A}_\tau)$ elements from $\G(\tilde{A}_\tau)$ for some computable function $f_2$.
By~\cite[Lemma 3.20]{Onn}, $g_1(\tilde{A}_\tau) \leq (\tilde{t}-r_\tau)(a\sqrt{r_\tau})^{r_\tau} \leq \tilde{t}(ar_\tau)^{O(r_\tau)} =: \zeta$, and thus $\zeta \cdot f_2(\tilde{A}_1, \dots, \tilde{A}_\tau)$ is an upper bound on $g_1(\tilde{A})$.

Observe that $A$ can be obtained from $\tilde{A}$ by repeating columns.
De Loera et al.~\cite[Corollary 3.7.2]{DHK} show that repeating columns does not increase the $\ell_1$-norm of Graver elements, and thus $g_1(A) \leq f_{\textrm{tf-norm}}(a, r_1, \dots, r_\tau) = \zeta \cdot f_2(\tilde{A}_1, \dots, \tilde{A}_\tau)$.
\epr
\end{accumulate}
\begin{lemma} \APXmark \label{lem:incidence_treewidth}
Let $X \in \N$.
Problem~\eqref{IP} with the additional constraint $\|\vex\|_1 \leq X$
can be solved in time $(aX)^{O(\tw_D(A))} \cdot n$, where $a = \max \{2, \|A\|_\infty\}$.
\end{lemma}

\ifthenelse{\value{Accumulate} = 1}{
\subparagraph{Proof sketch.} Lemma~\ref{lem:incidence_treewidth} is proved by reformulating the nonlinear constraint $\|\vex\|_1 \leq X$ by ``splitting'' each variable $x_i$ into two non-negative variables $x_i = x_i^+ - x_i^-$, imposing the constraint $\sum_{i=1}^n (x_i^+ + x_i^-) \leq X$, and showing that this does not increase $\tw_D(A)$ much; then, a recent dynamic programming algorithm of Ganian et al.~\cite[Theorem 6]{GOR} does the job.}
{}
\begin{accumulate}
\ifthenelse{\value{Accumulate} = 1}{
\begin{proof}[Proof of Lemma~\ref{lem:incidence_treewidth}]}
{\bpr}
The algorithm follows from a recent result of Ganian et al.~\cite{GOR}:
\begin{proposition}[Ganian et al.~{\cite[Theorem 6]{GOR}}] \label{prop:incidence_treewidth}
\eqref{IP} can be solved in time $\Gamma^{O(\tw_I(A))} \cdot n$,
where $\Gamma = \max_{\vex \in \Z^n \,:\, A \vex = \veb,\, \vel \leq \vex \leq \veu} \max_{i \in [n]} \left \|\sum_{j=1}^i A_j x_j \right \|_\infty$.
\end{proposition}
In other words, the parameter $\Gamma$ is bounding the largest number in absolute value which appears in any prefix sum of $A \vex$ for any feasible solution $\vex$.
Note that $\tw_I(A) \leq \tw_D(A) + 1$ by Lemma~\ref{lem:inc_prim_tw}; our proof applies for the more general case of incidence treewidth, but we only state Lemma~\ref{lem:incidence_treewidth} in terms of $\tw_D(A)$ for simplicity of exposition.

In order to use Proposition~\ref{prop:incidence_treewidth}, our only task is to formulate an auxiliary problem where we replace the nonlinear constraint $\|\vex\|_1 \leq X$ with a linear constraint.
This is easy by splitting every variable $x_i$ into its positive and negative part $x_i^+$ and $x_i^-$ which we force to be nonnegative by setting $x_i^+, x_i^- \geq 0$.
Correspondingly, every column $A_i$ of $A$ is now split into $A_i^+ = A_i$ and $A_i^- = -A_i$.
The bounds $l_i \leq x_i \leq u_i$ are rewritten to $l_i \leq x_i^+ - x_i^- \leq u_i$.
Finally, $\|\vex\|_1 \leq X$ in the original problem is equivalent to setting $\sum_{i=1}^n (x_i^+ + x_i^-) \leq X$, and we additionally set $0 \leq x_i^-, x_i^+ \leq X$.
It is easy to observe that this auxiliary problem has incidence treewidth at most $2\tw_I(A) + 2$: replace $x_i$ in each bag by $x_i^+, x_i^-$, add the constraint $\sum_{i=1}^n (x_i^+ + x_i^-) \leq X$ into each bag, and add one of $l \leq x_i^+ - x_i^-$ or $x_i^+ - x_i^- \leq u$ for at most one $i$ for each bag, perhaps for multiple copies of the original bag.
Moreover, by the fact that $\|\vex\|_1 \leq X$ and $a = \|A\|_\infty$, we have that $\Gamma \leq aX$.
The claim follows.
\epr
\end{accumulate}

\begin{proof}[Proof of Theorem~\ref{thm:treefold}]
Let $A = T^D(A_1, \dots, A_\tau)$ be a tree-fold matrix.
By Lemma~\ref{lem:treefold_norm} we have that $g_1(A) \leq f_{\textrm{tf-norm}}(a, r_1, \dots, r_\tau) =: M$
We show how to construct a $\Lambda$-Graver-best oracle.
Given an integer $\lambda \in \N$, solve
\mbox{
$
\min \{\lambda \vew \veh \mid A \veh = 0, \, \vel \leq \vex + \lambda \veh \leq \veu, \, \|\veh\|_1 \leq M\}
$
}
using Lemma~\ref{lem:incidence_treewidth}; clearly the result is a $\lambda$-Graver-best step.
Using a simple induction and the inductive construction of $A$, one gets that $A$ has $\tw_D(A) \leq \td_D(A) \leq r_1 + \cdots + r_\tau + 1$ and thus the oracle is realized in \FPT time.
Lemma~\ref{lem:lambda_gb_oracle} then yields a Graver-best oracle, which, combined with Theorem~\ref{thm:oracle}, finishes the proof.
\end{proof}

\subparagraph{$n$-fold ILP.}
A special case of tree-fold ILP is $n$-fold ILP, obtained by taking $T$ to be the star with $n$ leaves and $A = T^D(A_1, A_2)$, where $A_1 \in \Z^{r \times t}$ and $A_2 \in \Z^{s \times t}$.


\begin{proof}[Proof of Theorem~\ref{thm:nfolds}]
Before we apply Lemma~\ref{lem:incidence_treewidth}, we need to bound $g_1(A)$.
It follows from the proof of \cite[Lemma 6.1]{HOR} that there is a number $g(A) = \max_{\vev \in \G(A_1 \G(A_2))} \|\vev\|_1$ such that $g_1(A) \leq g(A) \cdot g_1(A_2)$.

Let $d_2 \leq (2a+1)^s$ be the number of distinct columns of $A_2$.
De Loera et al.~\cite{DHK} give a bound on $g_1(A)$ in terms of the number of distinct columns of a matrix $A$:
\begin{lemma}[$g_1(A)$ only depends on distinct columns~{\cite[Corollary 3.7.4]{DHK}}] \label{lem:rep_columns_bound}
Let $A \in \Z^{m \times n}$ be a matrix of rank $r$, let $d$ be the number of different columns in $A$, and let $a = \max \{2, \|A\|_\infty\}$. Then
$
g_1(A) \leq (d-r)(r+1)(\sqrt{m}a)^m \enspace .
$
\end{lemma}
Thus, $g_1(A_2) \leq (d_2-s)(s+1)(\sqrt{s}a)^s \leq (as)^{O(s)}$.
Let $G_2$ be a matrix whose columns are elements of $\G(A_2)$.
We have that $\|A_1G_2\|_\infty \leq a \cdot (as)^{O(s)} \leq (as)^{O(s)}$.
Moreover, since $A_1 G_2$ has $r$ rows, it has at most $d_1 = \left((as)^{O(s)} \right)^r = (as)^{O(rs)}$ distinct columns.
Again, by Lemma~\ref{lem:rep_columns_bound} we have that $g(A) = g_1(A_1 G_2) \leq (d_1 - r)(r+1)(\sqrt{r}as)^{O(s)})^r \leq (ars)^{O(rs)}$.
Combining, we get $g_1(A) \leq (ars)^{O(rs)} \cdot (as)^{O(s)} \leq (ars)^{O(rs)} =: M$.

We have $\tw_D(A) \leq r+s + 1$ and thus running the algorithm of Lemma~\ref{lem:incidence_treewidth} once takes time $\left( (ars)^{O(rs)} \right)^{r+s} nt \leq (ars)^{O(r^2 s + r s^2)} nt$ and finds the $\lambda$-Graver-best step.
Lemma~\ref{lem:lambda_gb_oracle} then yields a Graver-best oracle, which, combined with Theorem~\ref{thm:oracle}, finishes the proof.
\epr

\section{Primal and Dual Treedepth} \label{sec:treedepth}
\begin{accumulate}
\ifthenelse{\value{Accumulate} = 1}{
\section{Additions to Section~\ref{sec:treedepth}}
\setcounter{figure}{0}
}{}
\end{accumulate}

We prove Theorems~\ref{thm:primaltd} and~\ref{thm:dualtd} by showing that an ILP with bounded primal (dual) treedepth can be embedded into a multi-stage stochastic (tree-fold) ILP without increasing the parameters too much.
The precise notion of how one ILP is embedded in another is captured as follows.

\begin{definition}[Extended formulation]
Let $n' \geq n$, $m' \in \N$, $A \in \Z^{m \times n}$, $\veb \in \Z^m$, $\vel, \veu \in (\Z \cup \{\pm \infty\}^n$ and $A' \in \Z^{m^{\prime} \times n^{\prime}}$, $\veb' \in \Z^{m'}$, $\vel', \veu' \in (\Z \cup \{\pm \infty\})^{n'}$.
We say that $A'(\vex, \vey) = \veb', \vel' \leq (\vex, \vey) \leq \veu'$ is an \emph{extended formulation} of $A \vex = \veb, \vel \leq \vex \leq \veu$ if $\{\vex \mid A \vex = \veb, \vel \leq \vex \leq \veu\} = \{\vex \mid \exists \vey: A'(\vex, \vey) = \veb', \vel' \leq (\vex, \vey) \leq \veu'\}$.
\end{definition}

We note that from here on we always assume that if $\td_P(A) = k$ or $\td_D(A) = k$, then there is a \emph{tree} (not a forest) $F$ of height $k-1$ such that $G_P(A) \subseteq \cl(F)$ or $G_D(A) \subseteq \cl(F)$, respectively.
Otherwise $G_P(A)$ is not connected and each component corresponds to a subset of variables which defines an ILP that can be solved independently; similarly for $G_D(A)$.

\subsection{Primal Treedepth}
\begin{lemma}[Bounded primal treedepth $\Rightarrow$ multi-stage stochastic]
\label{lem:primaltd_uniformization}
Let $A, \veb, \vel$ and $\veu$ as in~\eqref{IP} be given, let $a = \max \{2, \|A\|_\infty\}$ and $\tau + 1 = \td_P(A)$.
Then there exists $C \in \Z^{m' \times n'}$, $\veb' \in \Z^{m'}$ and $\vel', \veu' \in (\Z \cup \{\pm \infty\})^{n'}$, $n' \leq n \tau$, $m' \leq (2a+1)^{\tau^2}$,  which define an integer program $C(\vex, \vey) = \veb', \vel' \leq (\vex, \vey) \leq \veu'$, which is an extended formulation of $A \vex = \veb, \vel \leq \vex \leq \veu$.
Moreover, there exist matrices $B_1, \dots, B_{\tau-1} \in \Z^{(2a+1)^{\tau^2} \times \tau}$ and $B_\tau \in \Z^{(2a+1)^{\tau^2} \times \left(\tau + (2a+1)^{\tau^2}\right)}$ and a tree $T$ such that $C = T^P(B_1, \dots, B_\tau)$ is a multi-stage stochastic constraint matrix, and all can be computed in time $f_{\textrm{P-embed}}(a, \td_P(A)) \cdot n^2$ for some computable function $f_{\textrm{P-embed}}$.
\end{lemma}

\bpr
Let $F$ be a rooted tree of height $\tau$ such that $G_P(A) \subseteq \cl(F)$ (recall that it can be computed in time $f_{\textrm{td}}(\td_P(A)) \cdot |V(G_P(A))|$).

\subparagraph{Step 1: Dummy columns.}

\begin{accumulate}
\begin{figure}[bt]
\centering
\includegraphics[scale=1]{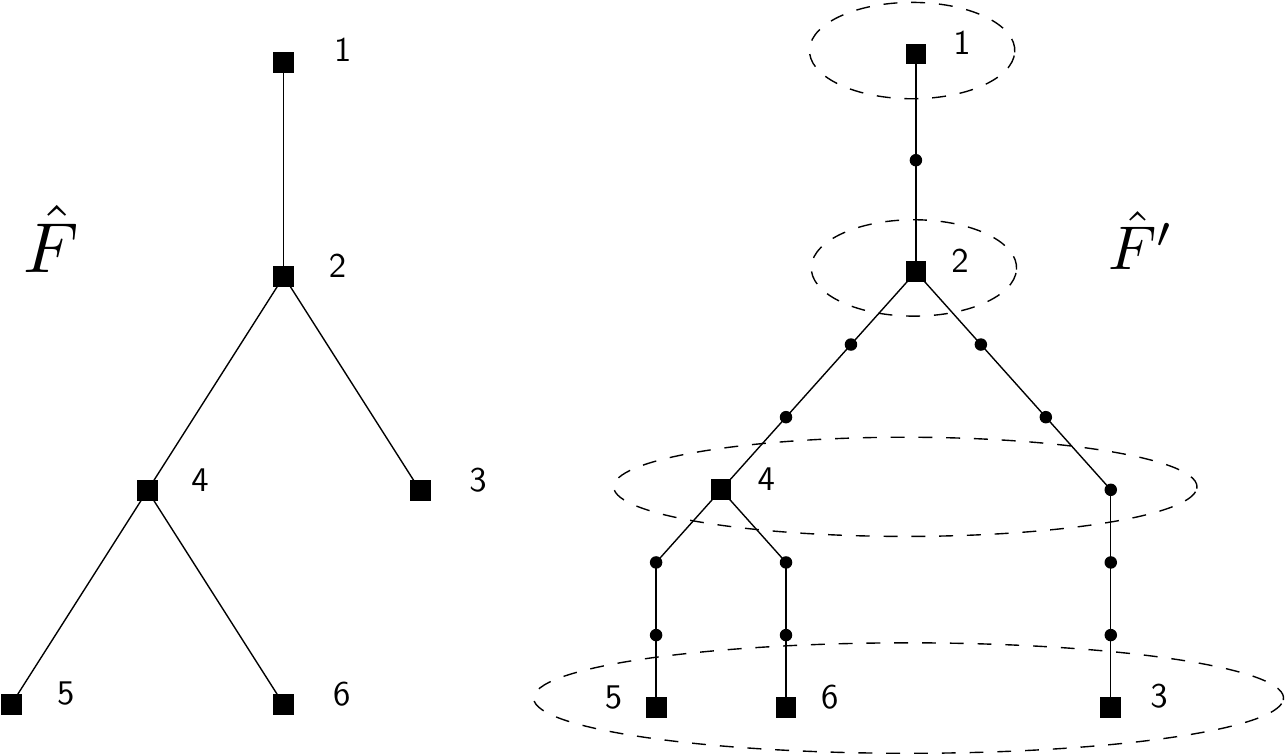}
\caption{Tree $\hat{F}'$ obtained from $\hat{F}$ by the transformation of step 1. We have $\tau = 3$, squares correspond to original variables; discs correspond to dummy variables; the frets are encircled.}
\label{fig:trees}
\end{figure}

\end{accumulate}

We make $F$ structured by adding dummy columns.
Observe that every root-leaf path is of length at most $\tau$ and thus contains at most $\tau$ branching vertices.
Unless $F$ is a path, obtain a matrix $A'$ from $A$ by inserting zero columns into $A$ in order to make the path between any two branching vertices of length $\tau$; a special case is the root which we force to be in distance $\tau-1$ from the closest branching vertex.
Set lower and upper bounds on the corresponding new variables to $0$.
A zero column is an isolated vertex in the primal graph and thus can be inserted to an arbitrary path of the tree $F$.
Moreover, if any leaf is at depth less than $\tau^2-1$, insert zero columns in the same way to make it be at depth exacty $\tau^2-1$.
Now there exists a rooted tree $F'$ of height $\tau^2-1$ such that $G_P(A') \subseteq \cl(F')$, all branching vertices are in distances $0, \tau-1, 2\tau-1, \dots, (\tau-1)\tau-1$ from the root, and all leaves are at depth exactly $\tau^2-1$.
We call all vertices at depth $0$ or $i \tau-1$, for $i \in [\tau]$, \emph{frets}, including the root and leaves.
For example, take the tree $\hat{F}$ in Figure~\ref{fig:trees}; the result $\hat{F}'$ of this procedure is depicted next to it.

\subparagraph{Step 2: Multi-stage stochastic extended formulation.}
Consider a root-leaf path $P$ in $F'$: its vertex set $V(P)$ corresponds to a certain subset of the columns of $A'$, with $|V(P)| \leq \tau^2$.
Furthermore, any row $\vea$ of $A'$ with $\suppo(\vea) \subseteq V(P)$ can be written as a vector $(\vea^1, \dots, \vea^\tau) \in \Z^{\tau^2}$ with $\vea^i \in \Z^\tau$ for each $i \in [\tau]$, and with $\|\vea\|_\infty \leq a$.
The \emph{bricks} $\vea^i$ correspond to segments between frets (including the end fret, i.e., the fret farthest from the root; segments adjacent to the root also contain the root).
Also, for any row $\vea$ of $A'$, there exists some root-leaf path $P$ such that $\suppo(\vea) \subseteq V(P)$.

This inspires the following construction: let $B \in \Z^{(2a+1)^{\tau^2} \times \tau^2}$ be the matrix whose columns are all the possible vectors $\vea \in \Z^{\tau^2}$ with $\|\vea\|_\infty \leq \|A\|_\infty$.
Let $B_i \in \Z^{(2a+1)^{\tau^2} \times \tau}$, for $i \in [\tau]$, be the submatrix of $B$ formed by rows $(i-1)\tau+1, \dots, i\tau$ and modify the last such submatrix $B_\tau$ by putting $B_\tau := (B_\tau \mid I)$ where $I \in \Z^{(2a+1)^{\tau^2} \times (2a+1)^{\tau^2}}$ is the identity matrix; the variables corresponding to columns of $I$ will play the role of slack variables.
Let $T$ be the tree of height $\tau$ obtained from $F'$ by contracting all paths between frets.

Now, let $C = T^P(B_1, \dots, B_\tau)$.
Obtain $\tilde{F}$ from $F'$ by appending a leaf to every leaf, and observe that $G_P(T^P(B_1, \dots, B_\tau)) \subseteq \cl(\tilde{F})$; the new leaves correspond to the slack variables in $B_\tau$.
Our goal now is to construct a right hand side vector $\veb'$ and lower and upper bounds $\vel', \veu'$ to enforce exactly the constraints present in $A\vex = \veb$.
For every root-leaf path $P$ in $F'$ there is a corresponding root-leaf path $\tilde{P}$ in $\tilde{F}$ such that $\tilde{P}$ is $P$ with an additional leaf.
Fix a root-leaf path $P$ in $F'$. 
For every row $\vea$ of $A'$ with $\suppo(\vea) \subseteq V(P)$ and right hand side $\beta$, there exists a unique row $\vecc$ of $C$ with $\suppo(\vecc) \subseteq V(\tilde{P})$ such that $\vecc = (\vea, 1)$, and we set the right hand side of row $\vecc$ to $\beta$.


For every row $\vecc$ of $C$ which was not considered in the previous paragraph, set the right hand side to $0$ and for the slack variable of this row set the lower bound to $-\infty$ and the upper bound to $\infty$.
Let us remark in passing that we are not limited by using the standard equality form of ILP: transforming an instance $A \vex \leq \veb$ into the standard equality form by adding slack variables only possibly increases the treedepth by $1$.
\epr

\subsection{Dual Treedepth}
\begin{lemma}[Bounded dual treedepth $\Rightarrow$ tree-fold] \APXmark
\label{lem:dualtd_uniformization}
Let $A, \veb, \vel$ and $\veu$ be as in~\eqref{IP}, $a = \max \{2, \|A\|_\infty\}$ and $\tau + 1 = \td_D(A)$.
Then there exists $D \in \Z^{m' \times n'}$, $\veb' \in \Z^{m'}$ and $\vel', \veu' \in (\Z \cup \{\pm \infty\})^{n'}$, $n' \leq nt$, $t \leq n$, $m' \leq m \cdot \tau$, which define an extended formulation of $A \vex = \veb, \vel \leq \vex \leq \veu$.
Moreover, there exist matrices $A_1, \dots, A_\tau \in \Z^{\tau \times t}$ and a tree $T$ such that $D = T^D(A_1, \dots, A_\tau)$ is a tree-fold constraint matrix, and all can be computed in time $f_{\textrm{D-embed}}(a, \td_D(A)) \cdot n^2$ for some computable function $f_\textrm{D-embed}$.
\end{lemma}

\ifthenelse{\value{Accumulate} = 1}{
\subparagraph{Proof sketch.} The proof is analogous to that of Lemma~\ref{lem:primaltd_uniformization}, except now paths in the tree $F$ correspond to collections of columns of $A$ instead of rows and thus there are no slack variables in the new tree-fold instance; however, we have to deal with duplicitous columns.
}
{}

\begin{accumulate}
\ifthenelse{\value{Accumulate} = 1}{
\begin{proof}[Proof of Lemma~\ref{lem:dualtd_uniformization}]}
{\bpr}
\ifthenelse{\value{Accumulate} = 0}{
The construction is analogous to the one in Lemma~\ref{lem:primaltd_uniformization}; the key differences are that
the paths in a tree $F$ ($G_D(A) \subseteq \cl(F)$) correspond to collections of columns of $A$ instead of rows and thus there are no slack variables in the tree-fold instance which we shall construct, however we have to deal with duplicitous columns.
}{}

Let $F$ be a tree of height $\tau$ such that $G_D(A) \subseteq \cl(F)$, let $A'$ be a matrix and $F'$ be a tree of height $\tau^2-1$ obtained as in the proof of Lemma~\ref{lem:primaltd_uniformization} (swapping rows and columns and noting that adding zero rows with zero right hand sides does not change the problem) such that $G_D(A') \subseteq \cl(F')$.
Similarly let $T$ be the tree obtained from $F'$ by contracting all paths between frets.

Let $\tilde{C} \in \Z^{\tau^2 \times (2a+1)^{\tau^2}}$ be the matrix of all columns of dimension $\tau^2$ with values from $[-a, a]$.
For every column $\vecc$ of $\tilde{C}$ let $m_\vecc = \max_{P: \text{ root-leaf path in } F'} |\{\vea \mid \vea = \vecc,\, \suppo(\vea) \subseteq V(P)\}|$ be the maximum multiplicity of columns identical to $\vecc$ over all root-leaf paths in $F'$.
Now let $C \in \Z^{\tau^2 \times t}$, $t = \sum_{\vecc \in C} m_\vecc$, be obtained from $\tilde{C}$ by taking each column $\vecc$ with multiplicity $m_{\vecc}$ and let $A_i \in \Z^{\tau \times t}$ for $i \in [\tau]$ be the matrix composed of rows $(i-1)\tau+1, \dots, i\tau$ of $C$.
Finally, let $D = T^D(A_1, \dots, A_\tau)$; we have $G_D(T^D(A_1, \dots, A_\tau)) \subseteq \cl(F')$.
It is now easy to see that we can construct an injective mapping $\varphi$ from the columns of $A'$ to the columns of $D$ such that $\varphi(\vea) = \ved$ satisfies $\vea = \ved$ and $\suppo(\vea) = \suppo(\ved) = V(P)$ for some root-leaf path $P$ in $F'$.

It remains to construct $\vel', \veu'$ and $\veb'$.
Let $\veb'$ be constructed from $\veb$ by inserting zeroes in place where zero rows were inserted when constructing $A'$ from $A$.
Let $\vea_i$ be the $i$-th column of $A'$ and $\ved_j = \varphi(\vea_i)$ be the $j$-th column of $D$.
Then set $l'_j = l_i$ and $u'_j = u_i$.
For every column $\ved_j$ of $D$ without a preimage in $\varphi$ set $l'_j = u'_j = 0$.

We show that $D \vex' = \veb', \vel' \leq \vex' \leq \veu'$ is an extended formulation of $A \vex = \veb, \vel \leq \vex \leq \veu$.
The columns of $D$ which do not have a preimage in $A'$ correspond to variables which are always $0$, so we can disregard them.
Thus the set of columns of $D$ corresponding to possibly non-zero variables is exactly the set of columns of $A'$.
Moreover, $A'$ has been derived from $A$ by inserting zero rows, which do not affect feasibility.
\epr
\end{accumulate}

\begin{accumulate}
\subsection{Avoiding Reductions}
We desire to highlight that it is in fact not necessary to reduce a bounded primal (dual) treedepth instance to a multi-stage stochastic (tree-fold) instance.
All that is needed are the bounds on $g_\infty(A)$ and $g_1(A)$, since then the algorithms we use (Lemma~\ref{lem:primal_treewidth} and Lemma~\ref{lem:incidence_treewidth}) already work on bounded treedepth instances.
To see that, it suffices to show that $g_\infty(A) \leq g_\infty(A')$ and $g_1(A) \leq g_1(A')$ for any extended formulation of~\eqref{IP} with a matrix $A'$, i.e., the norm bounds transfer.
These claims are an easy corollary of the following stronger lemma.

\begin{lemma}[Graver basis of extended formulations] \label{lem:graver_ef}
Let $A' (\vex, \vey) = \veb', \vel' \leq (\vex, \vey) \leq \veu'$ be an extended formulation of $A \vex = \veb, \vel \leq \vex \leq \veu$.
Then $\G(A) \subseteq H = \{\veg \mid \exists \veh: (\veg, \veh) \in \G(A')\}$.
\end{lemma}

To prove the lemma we need the notion of a test set.
A set $\T \subseteq \Z^n$ is a \emph{test set} for~\eqref{IP} if, for any non-optimal feasible solution $\vex$, there exists $\veg \in \T$ such that $\vex + \veg$ is feasible and $\vew (\vex+\veg) < \vew \vex$.
It is known that $\G(A)$ is an inclusion-wise minimal test set for all ILPs with the constraint matrix $A$~\cite[Lemma 3.3.2]{DHK}.

\begin{proof}[Proof of Lemma~\ref{lem:graver_ef}]
It suffices to show that $H$ is a test set.
Let $\vex$ be a non-optimal feasible solution in~\eqref{IP}.
By definition, there exists a $\vey$ such that $(\vex, \vey)$ is a non-optimal feasible solution in $A' (\vex, \vey) = \veb', \vel' \leq (\vex, \vey) \leq \veu'$.
By the test set property of $\G(A')$ we have that there exists $(\veg, \veh) \in \G(A')$ which is an augmenting step for $(\vex, \vey)$ and thus $\veg$ is an augmenting step for $\vex$.
\end{proof}
\end{accumulate}

\begin{accumulate}
\subsection{Hardness}

\ifthenelse{\value{Accumulate} = 0}{
Using a proof of Ganian et al.~\cite[Theorem 12]{GOR} we show that we cannot hope to relax the parameter $\td_D(A)$ to $\tw_D(A)$, even if $\|A\|_\infty$ was a constant.
\begin{lemma} \APXmark \label{lem:dualtd_nph}
\eqref{IP} is \NPh already when $\tw_D(A) = 3$, $\|A\|_\infty = 2$, and $\vew = \mathbf{0}$.
\end{lemma}}
{}

\ifthenelse{\value{Accumulate} = 1}{
\begin{proof}[Proof of Lemma~\ref{lem:dualtd_nph}]}
{\bpr}
The proof uses the same construction as Ganian et al.~\cite[Theorem 12]{GOR}; they only observe that it has bounded incidence treewidth, while we notice that it also has bounded dual treewidth.
Let $S = \{s_1, \dots, s_n\} \subseteq \N$, $s \in \N$, be an instance of \textsc{Subset Sum}, i.e., we are to decide whether there is a subset $S' \subseteq S$ such that $\sum_{s_i \in S'} s_i = s$.
We formulate this as an ILP feasibility problem.
Let $L = \max_{i \in [n]} \l s_i \r$.
Let $x_1, \dots, x_n$ be binary variables where $x_i$ encodes whether $s_i \in S'$.
Let $z_i, \dots, z_n$ be variables where $z_i$ has bounds $0 \leq z_i \leq s_i$ and we will enforce $(x_i = 1) \implies (z_i = s_i)$ and $(x_i = 0) \implies (z_i = 0)$.
The main trick needed to keep the coefficients small is the following.
Let $o_i^j$ be the $j$-th bit of the binary encoding of $s_i$.
We will introduce variables $y_i^j$, $j \in [0,L]$ and enforce that $(x_i = 1) \implies (y_i^j = o_i^j \cdot 2^j)$ and $(x_i = 0) \implies (y_i^j = 0)$; thus, we can then set $z_i = \sum_{i=0}^L o_i^j y_i^j$.
The constraints $A (\vex, \vey, \vez) = \veb$ are as follows:
\begin{align}
y_i^0 &= x_i & \forall i \in [n] \tag{$X_i$}\\
y_i^j &= 2y_i^{j-1} & \forall i \in [n], j \in [L] \tag{$Y_i^j$} \\
z_i &= \sum_{i=0}^L o_i^j y_i^j & \forall i \in [n] \tag{$Z_i$} \\
\sum_{i=1}^n z_i &= s & \enspace .\tag{$S$} 
\end{align}
Let us analyze the dual graph $G_D(A)$.
Its vertices are the constraints $X_i$, $Y_i^j$, $Z_i$ and $S$.
There are the following edges:
\begin{itemize}
\item Between $S$ and each $Z_i$,
\item between $Z_i$ and each $Y_i^j$, for all $i \in [n]$,
\item between each $Y_i^j$ and $Y_i^{j+1}$, for all $i \in [n]$ and $j \in [0,L-1]$, and,
\item between $X_i$ and $Y_i^0$.
\end{itemize}
We can construct a tree decomposition (in fact, a path decomposition) by consecutively taking the following segment of bags for each $i \in [n]$: \[\{S, Z_i, Y_i^0, X_i\}, \{S, Z_i, Y_i^0, Y_i^1\}, \{S, Z_i, Y_i^1, Y_i^2\}, \dots, \{S, Z_i, Y_i^{L-1}, Y_i^L\} \enspace .\]
Since each bag is of size $4$, the treewidth is $3$; we disregard lower and upper bounds here as they can be handled independently for each variable without increasing the treewidth.
Note that due to the long path $Y_i^0 - Y_i^1 - \cdots - Y_i^L$ we know that $\td_D(A) \geq \log L$, and \textsc{Subset Sum} is only hard for $L$ sufficiently large in $n$.
\epr
\end{accumulate}

\begin{accumulate}
The fact that~\eqref{IP} is \NPh already with $\tw_P(A)$, $\tw_D(A)$ and $\|A\|_\infty$ constant suggest that $g_\infty(A)$ and $g_1(A)$ should be large; otherwise we could apply Lemma~\ref{lem:primal_treewidth} or~\ref{lem:incidence_treewidth}, contradicting \textsf{P $\neq$ NP}.
Now we prove that indeed $g_\infty(A)$ (and thus also $g_1(A)$) can be large already for constant $\tw_P(A)$, $\tw_D(A)$ and $\|A\|_\infty$.

\begin{lemma} \label{lem:largeg1}
For each $n \geq 2$, $n \in \N$, there exists a matrix $A \in \Z^{(n-1) \times n}$ with $\tw_P(A) = \tw_D(A) = 1$, $\|A\|_\infty = 2$, and $2^{n-1} \leq g_\infty(A) \leq g_1(A)$.
\end{lemma}
\begin{proof}
Let
\[
A:=\left(
\begin{array}{ccccccc}
2      & -1  & 0 & \cdots & 0 & 0   \\
0      & 2       & -1 & \cdots & 0 & 0\\
0   & 0  & 2    &  \cdots & \vdots & \vdots\\
\vdots   &  \vdots & \vdots     &  \ddots & -1 & 0\\
0      & 0 & 0 &\cdots  &    2 & -1  \\
\end{array}
\right),
\]
be an $(n-1) \times n$ matrix.
The sequence $\{1,2\}, \{2,3\}, \dots, \{n-1,n\}$ forms a tree decomposition of $G_P(A)$ of width $1$; analogously $\{1,2\}, \{2,3\}, \dots, \{n-2,n-1\}$ is a tree decomposition of $G_D(A)$ of width $1$.
Observe that every $\vex \in \Z^n$ with $A \vex = \mathbf{0}$ must satisfy that $x_{i+1} = 2x_i$ for each $i \in [n-1]$.
Clearly  $\veh = (1, 2, 4, \dots, 2^n) \in \N^{n}$ satisfies $A \veh = \mathbf{0}$, and from the previous observation it immediatelly follows that there is no $\veh' \in \N^n$ with $\veh' \sqsubset \veh$, and thus $\veh \in \G(A)$.
\end{proof}
\end{accumulate}

%
%
%

%

%

\section{Outlook}
We highlight a few questions related to our results.
First, how far can Theorem~\ref{thm:oracle} be extended towards separable convex objectives?
The main missing piece seems to be a strongly polynomial algorithms for some classes of separable convex objectives similar to the algorithm of Tardos~\cite{Tar}.
Second, is $4$-block $n$-fold ILP \FPT or \Wh{1}?
The importance of this notorious problem is highlighted here: we give complete answers for parameter $\td_P(A)$ and $\td_D(A)$, but parameterization by $\td_I(A)$ is still wide open, and $4$-block $n$-fold is the simplest open case.
Third, we are interested in other strongly \FPT algorithms.

\bibliographystyle{plainurl}
\bibliography{psp}

\ifaccumulating{%
\newpage\appendix\sloppy
\appendix
\renewcommand\thefigure{\thesection.\arabic{figure}}
\accuprint
}

\end{document}